\pdfoutput=1
\documentclass[12pt]{scrartcl}

\usepackage[scaled=.90]{helvet}
\usepackage{mathptmx} 
\usepackage{courier}
\usepackage[ngerman,english]{babel}
\usepackage{dsfont}
\usepackage{amsmath}
\usepackage{amssymb}

\usepackage{enumerate, amssymb, amsmath}
\newcommand{\N}{\ensuremath{{\mathbb N}}}

\usepackage{tikz} 
\usepackage[T1]{fontenc}
\usepackage{lmodern}
\usepackage[figurewithin=section]{caption}

\usepackage{amssymb, amsmath }
\usepackage{amsthm}

\newtheorem{theorem}{Theorem}[section]

\newtheorem{lemma}[theorem]{Lemma} 
\newtheorem{proposition}[theorem]{Proposition}

\newtheorem{remark}[theorem]{Remark}

 \newcommand{\un}[1]{{\underline #1}}
 \newcommand{\ov}[1]{{\overline #1}}
\newcommand{\R}{\ensuremath{{\mathbb R}}}
\newcommand{\E}{\ensuremath{{\mathbb E}}}

\DeclareMathOperator*{\argmax}{argmax}

\DeclareMathOperator*{\sgn}{sgn}

\begin{document}

\title{Optimal decision under ambiguity for diffusion processes}
\author{S\"oren Christensen\thanks{Christian-Albrechts-Universit\"at, Mathematisches Seminar, Kiel, Germany, email: {christensen}@math.uni-kiel.de.}}
\date{\today}
\maketitle

%
\begin{abstract}
In this paper we consider stochastic optimization problems for an ambiguity averse decision maker who is uncertain about the parameters of the underlying process. In a first part we consider problems of optimal stopping under drift ambiguity for one-dimensional diffusion processes. Analogously to the case of ordinary optimal stopping problems for one-dimensional Brownian motions we reduce the problem to the geometric problem of finding the smallest majorant of the reward function in a two-parameter function space. In a second part we solve optimal stopping problems when the underlying process may crash down. These problems are reduced to one optimal stopping problem and one Dynkin game. Examples are discussed.
\end{abstract}

\vspace{.8cm}

\cite{PS}

\section{Introduction}
In most articles dealing with stochastic optimization problems one major assumption is that the decision maker has full knowledge of the parameter of the underlying stochastic process. This does not seem to be a realistic assumption in many real world situations. Therefore, different multiple prior models were studied in the economic literature in the last years. Here, we want to mention \cite{DE} and \cite{ES}, and refer to \cite{CR} for an economic discussion and further references. 

In this setting it is assumed that the decision maker deals with the uncertainty via a worst-case approach, that is, she optimizes her reward under the assumption that the ``market'' chooses the worst possible prior. This is a natural assumption, and we also want to pursue this approach.

A very important class of stochastic optimization problems is given by optimal stopping problems. These problems arise in many different fields, e.g., in pricing American-style options, in portfolio optimization, and in sequential statistics. Discrete time problems of optimal stopping in a multiple prior setting were first discussed in \cite{R} and analogous results to the classical ones were proved. In this setting a generalization of the classical best choice problem was treated in detail in \cite{CR2}. In continuous time the case of an underlying diffusion with uncertainty about the drift is of special interest. The general theory (including adjusted Hamilton-Jacobi-Bellman equations) is developed in \cite{CR}. Some explicit examples are given there, but no systematic way for finding an analytical solution is described. In \cite{A07} the case of monotonic reward functions for one-dimensional diffusion processes is considered. The restriction to monotonic reward functions simplifies the problem since only two different worst-case measures can arise. 

Another class of stochastic optimization problems under uncertainty was dealt with in a series of papers starting with \cite{KW}: Portfolio optimization problems are considered under the assumption that the underlying asset price process may crash down at a certain (unknown) time point. The decision maker is again considered to be ambiguity averse in the sense that she tries to choose the best possible stopping policy out of the worst possible realizations of the crash date. See \cite{KS} for an overview on existing results.

The aim of this article is to treat optimal stopping problems under uncertainty for underlying one-dimensional diffusion processes. These kinds of problems are of special interest since they arise in many situations and often allow for an explicit solution. 

The structure of this article is as follows:
In Section \ref{sec:drift} we first review some well-known facts about the solution of ordinary optimal stopping problems for an underlying Brownian motion. These problems can be solved graphically by characterizing the value function as the smallest concave majorant of the reward function. Then we treat the optimal stopping problem under ambiguity about the drift in a similar way: The result is that the value function can be characterized as the smallest majorant of the reward function in a two-parameter class of functions. The main tool is the use of generalized $r$-harmonic functions. After giving an example and characterizing the worst-case measure, we generalize the results to general one-dimensional diffusion processes. 

In Section \ref{sec:crashes} we introduce the optimal stopping problem under ambiguity about crashes of the underlying process in the spirit of \cite{KS}. In this situation the optimal strategy can be described by two easy strategies: One pre-crash and one post-crash strategy. These strategies can be found as solutions of a one-dimensional Dynkin game and an ordinary optimal stopping problem, which can both be solved using standard methods. We want to point out that this model is a natural situation where Dynkin games arise and the theory developed in the last years can be used fruitfully.
As an explicit example we study the valuation of American call-options in the model with crashes. Here, the post-crash strategy is the well-known threshold-strategy in the standard Black-Scholes setting. The pre-crash strategy is of the same type, but the optimal threshold is lower. 

\section{Optimal stopping under drift ambiguity}\label{sec:drift}

\subsection{Graphical solution of ordinary optimal stopping problems}\label{sec:graph_sol}
Problems of optimal stopping in continuous time are well-studied and the general theory is well-developed. Nonetheless, the explicit solution to such problems is often hard to find and the class of explicit examples is very limited. Most of them are generalizations of the following situation, that allows for an easy geometric solution:\\
Let $(W_t)_{t\geq 0}$ be a standard Brownian motion on a compact interval $[a,b]$ with absorbing boundary points $a$ and $b$. We consider the problem of optimal stopping given by the value function
\[v(x)=\sup_{\tau}\E_x(g(W_\tau)\mathds{1}_{\{\tau<\infty\}}),\mbox{~~~~$x\in[a,b]$,}\]
where the reward function $g:[a,b]\rightarrow[0,\infty)$ is continuous and the supremum is taken over all stopping times w.r.t. the natural filtration for $(W_t)_{t\geq0}$. Here and in the following, $\E_x$ denotes taking expectation for the process conditioned to start in $x$. In this case it is well-known that the value function $v$ can be characterized as the smallest concave majorant of $g$, see \cite{DyY}. This means that the problem of optimal stopping can be reduced to finding the smallest majorant of $g$ in an easy class of functions. For finding the smallest concave majorant of a function $g$ one only has to consider affine functions, i.e., for each fixed point $x\in[a,b]$ the value of the smallest concave majorant is given by
\[\inf\{h_{c,d}(x):c,d\in\R, h_{c,d}\geq g\},\]
where $h_{c,d}$ is an element of the two-parameter class of affine functions of the form $h_{c,d}(y)=cy+d$. This problem can be solved geometrically, see Figure \ref{fig:conc}. We want to remark that this problem is indeed a semi-infinite linear programming problem:
\begin{align*}
\min ! ~~~~~&cx+d\\
\mbox{s.t~~~~~}&cy+d\geq g(y)\mbox{~~~for all $y\in [a,b].$}
\end{align*}
This gives rise to an efficient method for solving these problems, which can be generalized in an appropriate way, see \cite{HS} for an analytical method and \cite{C} for a numerical point of view.
\begin{figure}[ht]
\begin{center}
\includegraphics[width=7cm]{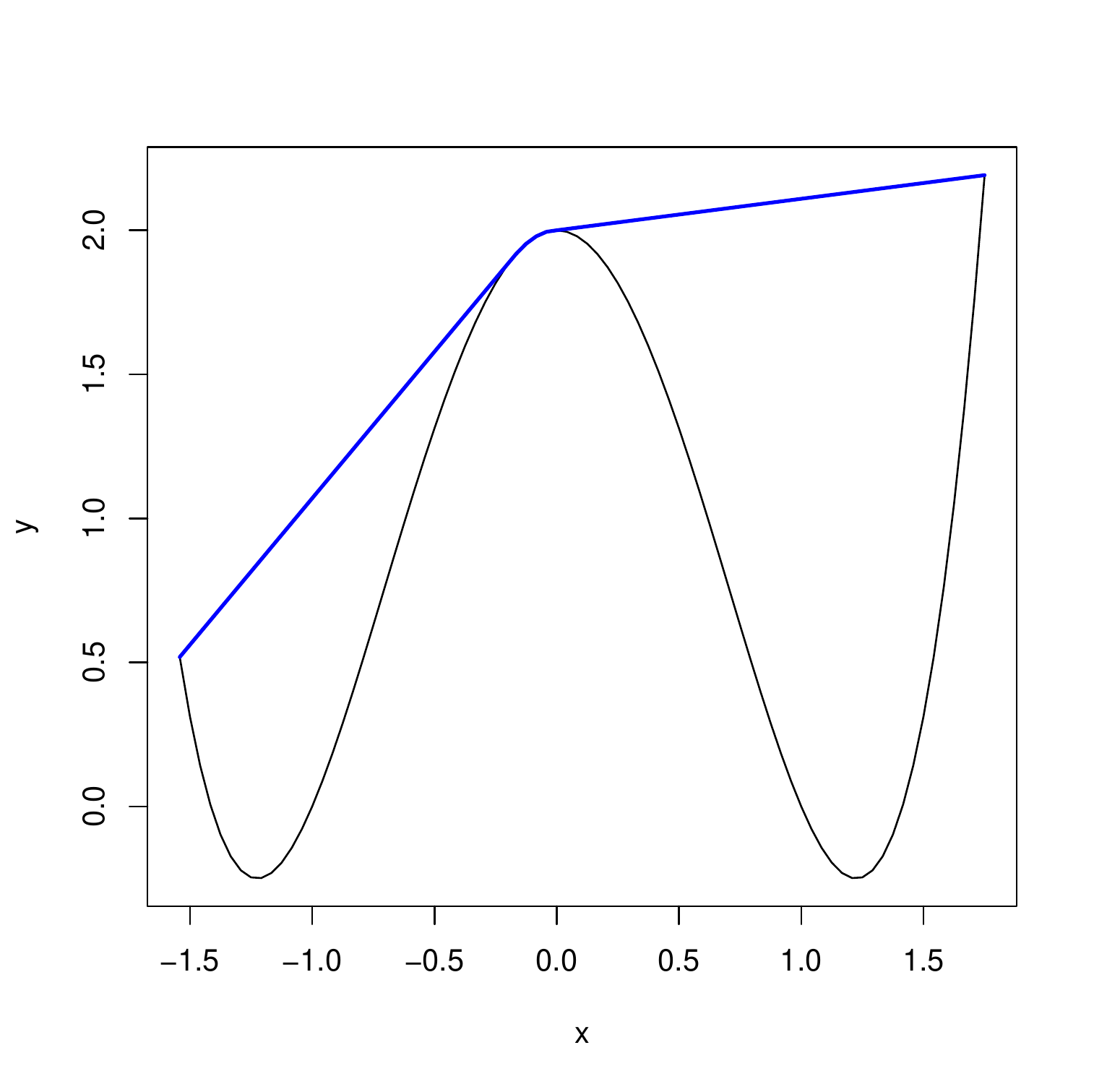}
\caption{Graph of a function $g$ (black) and its smallest concave majorant (blue)}\label{fig:conc}
\end{center}
\end{figure}
The example described above is important both for theory and applications of optimal stopping since by studying it one can obtain an intuition for more complex situations such as finite time horizon problems and multidimensional driving processes, where numerical methods have to be used in most situations of interest.\\


The goal of this section is to handle optimal stopping problems with drift ambiguity for diffusion processes similarly to the ordinary case discussed above. This gives rise to an easy to handle geometric method for solving optimal stopping problems under drift ambiguity explicitly. 

\subsection{Special Case: Brownian motion}\label{subsec:BM}
In the following we use the notation of \cite{CR}:
Let $(X_t)_{t\geq0}$  be a Brownian motion under the measure $Q$, fix $\kappa\geq0$ and denote by $\mathcal{P}^\kappa$ the set of all probability measures, that are equivalent to $Q$ with density process  of the form
\[\exp\left(\int_0^t\theta_sdX_s-1/2\int_0^t\theta^2_sds\right)\]
for a progressively measurable process $(\theta_t)_{t\geq0}$ with $|\theta_t|\leq \kappa$ for all $t\geq0$. We want to find the value function
\[v(x)=\sup_\tau\inf_{P\in\mathcal{P}^\kappa}\E_x^P(e^{-r\tau}g(X_\tau)\mathds{1}_{\{\tau<\infty\}})\]
for some fixed discounting rate $r>0$ and a measurable reward function $g:\R\rightarrow[0,\infty)$, where $\E_x^P$ means taking expectation under the measure $P$ when the process is started in $x$. Instead of taking affine functions as in Subsection \ref{sec:graph_sol} we construct another class of appropriate functions based on the minimal $r$-harmonic functions (introduced below) for the Brownian motion with drift $-\kappa$ resp. $\kappa$ as follows:\\
Denote the roots of the equation
\[1/2z^2-\kappa z-r=0\]
by $\alpha_1<0<\alpha_2$ and the roots of
\[1/2z^2+\kappa z-r=0\]
by $\beta_1<0<\beta_2$. Then $e^{\alpha_ix},i=1,2,$ are the minimal $r$-harmonic functions for a Brownian motion with drift $-\kappa$, and $e^{\beta_ix},i=1,2,$ the corresponding functions for a Brownian motion with drift $\kappa$. Note that $\beta_1\leq\alpha_1\leq0\leq\beta_2\leq\alpha_2$ and $\beta_1=-\alpha_2$ and $\beta_2=-\alpha_1$. For all $c\in \R$ define the functions $h_c:\R\rightarrow[0,\infty)$ via
\[h_c(x)=\begin{cases}
  \frac{\alpha_2}{\alpha_2-\alpha_1}e^{\alpha_1(x-c)}-\frac{\alpha_1}{\alpha_2-\alpha_1}e^{\alpha_2(x-c)},  & \text{if}~~x> c\\
  \frac{\beta_2}{\beta_2-\beta_1}e^{\beta_1(x-c)}-\frac{\beta_1}{\beta_2-\beta_1}e^{\beta_2(x-c)},  & \text{if}~~x\leq c,
\end{cases}\]
and
\[h_\infty(x)=e^{\beta_1x},~~~h_{-\infty}(x)=e^{\alpha_2x}.\]
For $c\in\R$, the function $h_c$ is constructed by smoothly merging $r$-harmonic functions for the Brownian motion with drift $\kappa$ (for $x\leq c$) and $-\kappa$ (for $x> c$) at their minimum in $c$. By taking derivatives and taking into account that $\beta_1=-\alpha_2$ and $\beta_2=-\alpha_1$, one sees that the function $h_c$ is indeed $C^2$. \\
The set $\{\lambda h_c:c\in[-\infty,\infty],\lambda\geq 0\}$ does not form a convex cone for $\kappa>0$. This is the main difference compared to the case without drift ambiguity. Therefore, the  standard techniques for optimal stopping are not applicable immediately. Nonetheless, this leads to the right $\mathcal{P}^\kappa$-supermartingales to work with:

\begin{lemma}\label{lem:harmonic}
\begin{enumerate}[(i)]
\item For all $a,b,x\in\R$ with $a\leq x\leq b$, $c\in[-\infty,\infty]$, $P\in\mathcal{P}^\kappa$ and $\tau=\inf\{t\geq 0:X_t\not\in[a,b]\}$ it holds that 
\[\E_x^{P}(e^{-r\tau}h_c(X_\tau)\mathds{1}_{\{\tau<\infty\}})\geq h_c(x)\mbox{~~and~~}\E_x^{P_c}(e^{-r\tau}h_c(X_\tau)\mathds{1}_{\{\tau<\infty\}})=h_c(x),\]
where the measure $P_c$ is such that 
\[dX_t=-\kappa \sgn(X_t-c)dt+dW^c_t\]
for a Brownian motion $W^c$ under $P_c$.
\item For all $c\in[-\infty,\infty]$ and all stopping times $\tau$ it holds that 
\[\E_x^{P_c}(e^{-r\tau}h_c(X_\tau)\mathds{1}_{\{\tau<\infty\}})\leq h_c(x).\]
\item For all $a,b,x\in\R$ with $a<x<b$, $P\in\mathcal{P}^\kappa$, and $\tau_a=\inf\{t\geq 0:X_t=a\}$, $\tau_b=\inf\{t\geq 0:X_t=b\}$ it holds that 
\begin{align*}
&\E_x^{P}(e^{-r\tau_a}h_\infty(X_{\tau_a})\mathds{1}_{\{\tau_a<\infty\}})\geq h_\infty(x),\\
&\E_x^{P_\infty}(e^{-r\tau_a}h_\infty(X_{\tau_a})\mathds{1}_{\{\tau_a<\infty\}})=h_\infty(x),
\end{align*}
and
\begin{align*}
&\E_x^{P}(e^{-r\tau_b}h_{-\infty}(X_{\tau_b})\mathds{1}_{\{\tau_b<\infty\}})\geq h_{-\infty}(x),\\
&\E_x^{P_{-\infty}}(e^{-r\tau_b}h_{-\infty}(X_{\tau_b})\mathds{1}_{\{\tau_b<\infty\}})=h_{-\infty}(x).
\end{align*}
\end{enumerate}
\end{lemma}
\begin{proof}
\begin{enumerate}[(i)]
\item For $P\in\mathcal{P}$ with density process $\theta$, by Girsanov's theorem, we may write
\[X_t=W_t^P+\int_0^t\theta_sds,\]
where $W^P$ is a Brownian motion under $P$. Since $h_c\in C^2$ we can apply It\^o's lemma and obtain
\[dh_c(X_t)=h_c'(X_t)dW^P_t+(h_c'(X_t)\theta_t+1/2h_c''(X_t))dt.\]
By construction of $h_c$, it holds that
\[1/2h_c''(X_t)-\kappa \sgn(X_t-c)h_c'(X_t)-rh_c(X_t)=0,\]
hence
\begin{align*}
e^{-rt}h_c(X_t)=&h_c(X_0)+\int_0^t e^{-ru}(\kappa\sgn(X_u-c)+\theta_u)h_c'(X_u)du\\
&+\int_0^t e^{-ru}h_c'(X_u)dW^P_u.
\end{align*}
Noting that $(\kappa\sgn(X_u-c)+\theta_u)\geq 0$ iff $h_c'(X_u)\geq 0$, we obtain that the process $(e^{-r(t\wedge \tau)}h_c(X_{t\wedge \tau}))_{t\geq0}$ is a bounded $P$-submartingale. Therefore, by the optional sampling theorem,
\[\E_x^P(e^{-r\tau}h_c(X_{\tau}))\geq \E_x^P(h_c(X_{0}))=h_c(x).\]
Under $P^c$ we see that $(e^{-r(t\wedge \tau)}h_c(X_{t\wedge \tau}))_{t\geq0}$ is actually a local martingale that is bounded. Therefore, the optional sampling theorem yields equality.
\item By the calculation in $(i)$ the process $(e^{-rt}h_c(X_t))_{t\geq0}$ is a positive local $P^c$-martingale, i.e. also a $P^c$-supermartingale. The optional sampling theorem for non-negative supermartingales is applicable.
\item By noting that $h_\infty$ is decreasing and $h_{-\infty}$ is increasing the same arguments as in $(i)$ apply.
\end{enumerate}
\end{proof}

The following theorem shows that the geometric solution described in Subsection \ref{sec:graph_sol} can indeed be generalized to the drift ambiguity case. Moreover, we give a characterization of the optimal stopping set as maximum point of explicitly given functions.
\begin{theorem}\label{thm:main_brownian}
\begin{enumerate}[(i)]
\item
 It holds that
\[v(x)=\inf\{\lambda h_{c}(x):c\in[-\infty,\infty],\lambda \in[0,\infty], \lambda h_{c}\geq g\}\mbox{~~~~~for all $x\in\R$}.\]
Furthermore, the infimum in $c$ is indeed a minimum.
\item
 A point $x\in\R$ is in the optimal stopping set $\{y:v(y)=g(y)\}$ if and only if there exists $c\in[-\infty,\infty]$ such that
\[x\in\argmax \frac{g}{h_c}.\]
\end{enumerate}
\end{theorem}
\begin{proof}
For each $x\in\R$, $c\in[-\infty,\infty]$ and each stopping time $\tau$ we obtain using Lemma \ref{lem:harmonic} (ii)
\begin{align*}
\inf_P\E_x^P(e^{-r\tau}g(X_\tau)\mathds{1}_{\{\tau<\infty\}})&=\inf_P\E_x^P\left(e^{-r\tau}h_c(X_\tau)\frac{g}{h_c}(X_\tau)\mathds{1}_{\{\tau<\infty\}}\right)\\
&\leq \sup\left(\frac{g}{h_c}\right)\inf_P\E_x^P(e^{-r\tau}h_c(X_\tau)\mathds{1}_{\{\tau<\infty\}})\\
&\leq \sup\left(\frac{g}{h_c}\right)\E_x^{P_c}(e^{-r\tau}h_c(X_\tau)\mathds{1}_{\{\tau<\infty\}})\\
&\leq \sup\left(\frac{g}{h_c}\right)h_c(x).
\end{align*}
Since $\lambda h_c\geq g$ holds if and only if $\lambda\geq \sup\left(\frac{g}{h_c}\right)$ we obtain that 
\[v(x)\leq \inf\{\lambda h_{c}(x):c\in[-\infty,\infty],\lambda \geq 0, \lambda h_{c}\geq g\}.\]
For the other inequality consider the following cases:\\
\textit{Case 1:}
\[\sup_{y\in\R}\frac{g(y)}{h_\infty(y)}=\sup_{y\leq x}\frac{g(y)}{h_\infty(y)}.\]
Take a sequence $(y_n)_{n\in\N}$ with $y_n\leq x$ such that ${g(y_n)}/{h_\infty(y_n)}\rightarrow\sup_{y\in\R}\frac{g(y)}{h_\infty(y)}$. Then for $\tau_n=\inf\{t\geq 0: X_t=y_n\}$ using Lemma \ref{lem:harmonic} (iii) we obtain
\begin{align*}
v(x)&\geq \inf_P\E_x^P(e^{-r\tau_n}g(X_{\tau_n})\mathds{1}_{\{\tau_n<\infty\}})\\
&=\inf_P\E_x^P(e^{-r\tau_n}h_\infty(X_{\tau_n})\frac{g}{h_\infty}(X_{\tau_n})\mathds{1}_{\{\tau_n<\infty\}})\\
&=\frac{g}{h_\infty}(y_n)E_x^{P_\infty}(e^{-r\tau_n}h_\infty(X_{\tau_n})\mathds{1}_{\{\tau_n<\infty\}})\\
&=\frac{g}{h_\infty}(y_n)h_\infty(x)\\
&\rightarrow \sup_{y\in\R}\frac{g(y)}{h_\infty(y)} h_\infty(x)~~~~~\mbox{for $n\rightarrow\infty$.} 
\end{align*}
Therefore, $v(x)\geq \inf\{\lambda h_{c}(x):c\in[-\infty,\infty],\lambda\in[0,\infty], \lambda h_{c}\geq g\}$.\\
Moreover, if $x$ is in the stopping set, i.e. $v(x)=g(x)$, then we see that $g(x)/h_\infty(x)=\sup_{y\in\R}\frac{g(y)}{h_\infty(y)}$, i.e. $x$ is a maximum point of the function ${g}/h_{\infty}$, i.e. $(ii)$.\\
\textit{Case 2:} The case $\sup_{y\in\R}{g(y)}/{h_{-\infty}(y)}=\sup_{y\geq x}{g(y)}/{h_{-\infty}(y)}$ can be handled the same way.\\
\textit{Case 3:}
\begin{align*}
\sup_{y\leq x}\frac{g(y)}{h_\infty(y)}>\sup_{y\geq x}\frac{g(y)}{h_\infty(y)}~~~\mbox{and}~~~\sup_{y\leq x}\frac{g(y)}{h_{-\infty}(y)}<\sup_{y\geq x}\frac{g(y)}{h_{-\infty}(y)}.
\end{align*}
First we show that there exists $c^*\in\R$ such that
\[\sup_{y\leq x}\frac{g(y)}{h_{c^*}(y)}=\sup_{y\geq x}\frac{g(y)}{h_{c^*}(y)}:\]
To this end, write
\begin{align*}
h_{c,1}(y)&=\frac{\alpha_2}{\alpha_2-\alpha_1}e^{\alpha_1(y-c)}-\frac{\alpha_1}{\alpha_2-\alpha_1}e^{\alpha_2(y-c)},\\
h_{c,2}(y)&=\frac{\beta_2}{\beta_2-\beta_1}e^{\beta_1(y-c)}-\frac{\beta_1}{\beta_2-\beta_1}e^{\beta_2(y-c)}.
\end{align*}
By construction of $h_c$ it holds that $h_c=\min(h_{c,1},h_{c,2})$. Therefore,
\begin{align*}
\sup_{y\leq x}\frac{g(y)}{h_c(y)}&=\left[\inf_{y\leq x}\left(\min\left(\frac{h_{c,1}(y)}{g(y)},\frac{h_{c,2}(y)}{g(y)}\right)\right)\right]^{-1}\\
=\bigg[\min\bigg\{&e^{-\alpha_2 c}\inf_{y\leq x}\left(\frac{\alpha_2}{\alpha_2-\alpha_1}\frac{e^{\alpha_1 y}}{g(y)}e^{(\alpha_2-\alpha_1)c}+\frac{\alpha_1}{\alpha_2-\alpha_1}\frac{e^{\alpha_2 y}}{g(y)}\right),\\
& e^{-\beta_2 c}\inf_{y\leq x}\left(\frac{\beta_2}{\beta_2-\beta_1}\frac{e^{\beta_1 y}}{g(y)}e^{(\beta_2-\beta_1)c}+\frac{\beta_1}{\beta_2-\beta_1}\frac{e^{\beta_2 y}}{g(y)}\right)\bigg\}\bigg]^{-1}
\end{align*}
Since the functions
\[z\mapsto \inf_{y\leq x}\left(\frac{\alpha_2}{\alpha_2-\alpha_1}\frac{e^{\alpha_1 y}}{g(y)}z+\frac{\alpha_1}{\alpha_2-\alpha_1}\frac{e^{\alpha_2 y}}{g(y)}\right)\]
and 
\[z\mapsto\inf_{y\leq x}\left(\frac{\beta_2}{\beta_2-\beta_1}\frac{e^{\beta_1 y}}{g(y)}z+\frac{\beta_1}{\beta_2-\beta_1}\frac{e^{\beta_2 y}}{g(y)}\right)\]
are continuous as concave functions, we obtain that the function $c\mapsto\sup_{y\leq x}\frac{g(y)}{h_c(y)}$ is continuous. By the same argument, the function $c\mapsto\sup_{y\geq x}\frac{g(y)}{h_c(y)}$ is also continuous. By the intermediate value theorem applied to the function
\[c\mapsto\sup_{y\leq x}\left(\frac{g(y)}{h_c(y)}\right)-\sup_{y\geq x}\left(\frac{g(y)}{h_c(y)}\right)\]
there exists $c^*$ with $\sup_{y\leq x}\frac{g(y)}{h_{c^*}(y)}=\sup_{y\geq x}\frac{g(y)}{h_{c^*}(y)}$ as desired.\\
Now take sequences $(y_n)_{n\in\N}$ and $(z_n)_{n\in\N}$ with $y_n\leq x\leq z_n$ such that 
\[\sup_{y\leq x}\frac{g(y)}{h_{c^*}(y)}=\lim_{n\rightarrow\infty}\frac{g(y_n)}{h_{c^*}(y_n)}=\lim_{n\rightarrow\infty}\frac{g(z_n)}{h_{c^*}(z_n)}=\sup_{y\geq x}\frac{g(y)}{h_{c^*}(y)}.\]
Using $\tau_n=\inf\{t\geq0:X_t\not\in[y_n,z_n]\}$ we obtain by Lemma \ref{lem:harmonic} (i)
\begin{align*}
v(x)&\geq \inf_P\E_x^P(e^{-r\tau_n}h_{c^*}(X_{\tau_n})\frac{g}{h_{c^*}}(X_{\tau_n})\mathds{1}_{\{\tau_n<\infty\}})\\
&\geq \left(\frac{g}{h_{c^*}}(y_n)\wedge \frac{g}{h_{c^*}}(z_n)\right)\inf_P\E_x^P(e^{-r\tau_n}h_{c^*}(X_{\tau_n})\mathds{1}_{\{\tau_n<\infty\}})\\
&=\left(\frac{g}{h_{c^*}}(y_n)\wedge \frac{g}{h_{c^*}}(z_n)\right)h_{c^*}(x)\rightarrow \sup\left(\frac{g}{h_{c^*}}\right)h_{c^*}(x).
\end{align*}
This yields the result $(i)$. As above we furthermore see that if $x$ is in the optimal stopping set, then it is a maximum point of $g/h_{c^*}$, i.e. $(ii)$.
\end{proof}
\begin{remark}
\begin{enumerate}
\item We would like to emphasize that we do not need any continuity assumptions on $g$. This is remarkable, because even for the easy case described at the beginning of this section most standard techniques do not lead to such a general result.
\item The previous proof is inspired by the ideas first described in \cite{BL}. It seems that other standard methods for dealing with optimal stopping problems for diffusions without drift ambiguity (such as Martin boundary theory as in \cite{sa}, generalized concavity methods as in \cite{dk}, or linear programming arguments as in \cite{HS}) are not applicable with minor modifications due to the nonlinear structure coming from drift ambiguity. A characterization of the optimal stopping points as in Theorem \ref{thm:main_brownian} (ii) for the problem without ambiguity can be found in \cite{CI2}.
\end{enumerate}
\end{remark}

\subsection{Worst-case prior}\label{subsec:worst_case}
Theorem \ref{thm:main_brownian} leads to the value of the optimal stopping problem with drift ambiguity and also provides an easy way to find the optimal stopping time. Another important topic is to determine the worst-case measure for a process started in a point $x$, i.e. we would like to determine the measure $P$ such that $v(x)=\sup_\tau \E^P_x(e^{-r\tau}g(X_\tau)\mathds{1}_{\{\tau<\infty\}})$. Using the results described above the worst-case measure can also be found immediately: 
\begin{theorem}\label{thm:worst_case_brown}
Let $x\in \R$ and let $c$ be a minimizer as in Theorem \ref{thm:main_brownian} (i). Then $P^c$ is a worst-case measure for the process started in $x$.
\end{theorem}

\begin{proof}
This is immediate from the proof of Theorem \ref{thm:main_brownian}.
\end{proof}


\subsection{Example: American straddle in the Bachelier market}
Because it is easy and instructive we consider the example discussed in \cite{CR} in the light of our method: \\
We consider a variant of the American straddle option in a Bachelier market model as follows: As a driving process we consider a standard Brownian motion under $P^0$ with reward function $g(x)=|x|$. Our aim is to find the value in 0 of the optimal stopping problem 
\[\sup_\tau \min_{P\in\mathcal{P}^\kappa}\E^P(e^{-r\tau}|X_\tau|\mathds{1}_{\{\tau<\infty\}}).\]
Using Theorem \ref{thm:main_brownian} we have to find the majorant of $|\cdot|$ in the set
\[\{\lambda h_{c}:c\in[-\infty,\infty],\lambda \in[0,\infty], \lambda h_{c}\geq g\}.\]
One immediately sees that if $\lambda h_c(\cdot)\geq |\cdot|$, then $\lambda h_0(\cdot)\geq |\cdot|$ and furthermore $\lambda h_0(0)\leq \lambda h_c(0)$. Therefore, we only have to consider majorants of $|\cdot|$ in the set
\[\{\lambda h_{0}:\lambda \in[0,\infty], \lambda h_{0}\geq g\}.\]
This one-dimensional problem can be solved immediately. For $\lambda=\max(|\cdot|/h_0(\cdot))$ one obtains $v(0)=\lambda h_0(0)$. \\
In fact, if $-b,b$ denote the maximum points of $|\cdot|/h_0(\cdot)$ we obtain that $v(x)=\lambda h_0(x)$ for $x\in[-b,b]$. Moreover, for $x\not\in[-b,b]$ one immediately sees that there exists $c\in\R$ such that $x$ is a maximum point of $|\cdot|/h_c(\cdot)$ and we obtain
\[v(x)=\begin{cases}
  \lambda h_0(x),  & \text{if}~~x\in[-b,b]\\
  |x|,  & \text{else}.
\end{cases}\]
Moreover, the worst-case measure is $P_0$, i.e. the process $X$ has positive drift $\kappa$ on $(-\infty,0)$ and drift $-\kappa$ on $[0,\infty)$.

\subsection{General diffusion processes}\label{subsec:general}
The results obtained before can be generalized to general one-dimensional diffusion processes. The only problem is to choose appropriate functions $h_c$ carefully. After these functions are constructed the same arguments as in the previous subsections work.\\
Let $(X_t)_{t\geq0}$ be a regular one-dimensional diffusion process on some interval $I$ with boundary points $a<b,a,b\in[-\infty,\infty]$, that is characterized by its generator
\[A=\frac{1}{2}\sigma^2(x)\frac{d^2}{dx^2}+\mu(x)\frac{d}{dx}\]
for some continuous functions $\sigma>0,\mu$. For convenience we furthermore assume that the boundary points $a, b$ of $I$ are natural. For a generalization to other boundary behaviors see the discussion in \cite[Section 6]{BL00}. Again denote by $\mathcal{P}^\kappa$ the set of all probability measures, that are equivalent to $Q$ with density process  of the form
\[\exp\left(\int_0^t\theta_sdX_s-1/2\int_0^t\theta^2_sds\right)\]
for a progressively measurable process $(\theta_t)_{t\geq0}$ with $|\theta_t|\leq \kappa$ for all $t\geq0$. We denote the fundamental solutions of the equation
\[\frac{1}{2}\sigma^2(x)\frac{d^2}{dx^2}\psi+(\mu(x)+\kappa)\frac{d}{dx}\psi=r\psi\]
by $\psi_+^\kappa$ resp. $\psi_-^\kappa$ for the increasing resp. decreasing positive solution, cf. \cite[II.10]{BS} for a discussion and further references. Analogously, denote the fundamental solutions of 
\[\frac{1}{2}\sigma^2(x)\frac{d^2}{dx^2}\psi+(\mu(x)-\kappa)\frac{d}{dx}\psi=r\psi\]
by $\psi_+^{-\kappa}$ resp. $\psi_-^{-\kappa}$.
Note that for each positive solution $\psi$ of one of the above ODEs it holds that 
\begin{equation}\label{eq:sec_der}
\frac{d^2}{dx^2}\psi(x)=\frac{-(\mu(x)\pm\kappa) \frac{d}{dx}\psi(x)+r\psi(x)}{\frac{1}{2}\sigma^2(x)},
\end{equation}
hence all extremal points are minima, so that $\psi$ has at most one minimum. Therefore, for each $s\in (0,1)$ the function $\psi=s\psi_+^{\pm\kappa}+(1-s)\psi_-^{\pm\kappa}$ has a unique minimum point and each $c\in I$ arises as such a minimum point. Therefore, for each $c\in (a,b)$ we can find constants $\gamma_1,...,\gamma_4$ such that the function
\[h_c:E\rightarrow\R,x\mapsto\begin{cases}
  \gamma_1\psi^\kappa_+(x)+\gamma_2\psi^{\kappa}_-(x),  & \text{if}~~x\leq c\\
  \gamma_3\psi^{-\kappa}_+(x)+\gamma_4\psi^{-\kappa}_-(x),  & \text{if}~~x>c
\end{cases}\]
is $C^1$ with a unique minimum point in $c$ with the standardization $h_c(c)=1$. More explicitly, $\gamma_1,...,\gamma_4$ are given by
\begin{equation}\label{eq:gamma}
\gamma_1=\frac{{\psi_-^{\kappa}}'(c)}{D^\kappa(c)},\; \gamma_2=\frac{-{\psi_+^{\kappa}}'(c)}{D^\kappa(c)},\;\gamma_3=\frac{{\psi_-^{-\kappa}}'(c)}{D^{-\kappa}(c)},\; \gamma_4=\frac{-{\psi_+^{-\kappa}}'(c)}{D^{-\kappa}(c)},
\end{equation}
where 
\[D^{\pm\kappa}(c)=\psi_+^{\pm\kappa}(c){\psi_-^{\pm\kappa}}'(c)-\psi_-^{\pm\kappa}(c){\psi_+^{\pm\kappa}}'(c).\]
Furthermore, write $h_a=\psi^{-\kappa}_+$ and $h_b=\psi^{\kappa}_-$. First, we show that the functions $h_c$ are always $C^2$.

\begin{lemma}
For each $c\in[a,b]$, the function $h_c$ is $C^2$. 
\end{lemma}

\begin{proof}
For $c\in\{a,b\}$ the claim obviously holds. Let $c\in(a,b)$. We only have to prove that $h_c''(c-)=h_c''(c+)$. Using equation \eqref{eq:sec_der}, we obtain that
\begin{align*}
h_c''(c-)&=\gamma_1{\psi^\kappa_+}''(c)+\gamma_2{\psi_-^{\kappa}}''(c)\\
&=\frac{-(\mu(c)+\kappa)}{\frac{1}{2}\sigma^2(c)}(\gamma_1{\psi_+^\kappa}'(c)+\gamma_2{\psi_-^\kappa}'(c))+\frac{r}{\frac{1}{2}\sigma^2(c)}(\gamma_1{\psi_+^\kappa}(c)+\gamma_2{\psi_-^\kappa}(c))\\
&=\frac{-(\mu(c)+\kappa)}{\frac{1}{2}\sigma^2(c)}h_c'(c)+\frac{r}{\frac{1}{2}\sigma^2(c)}h_c(c).
\end{align*}
By the choice of $\gamma_1,\gamma_2$, we obtain
\begin{align*}
h_c''(c-)&=\frac{r}{\frac{1}{2}\sigma^2(c)}
\end{align*}
and analogously 
\begin{align*}
h_c''(c+)&=\frac{r}{\frac{1}{2}\sigma^2(c)}.
\end{align*}
This proves the claim.
\end{proof}

Now all the arguments given in Subsection \ref{subsec:BM} and \ref{subsec:worst_case} apply and we again obtain the following results (compare Theorem \ref{thm:main_brownian} and Theorem \ref{thm:worst_case_brown}):
 
\begin{theorem}\label{thm:main_diffusion}
\begin{enumerate}[(i)]
\item
 It holds that
\[v(x)=\inf\{\lambda h_{c}(x):c\in [a,b],\lambda \in[0,\infty], \lambda h_{c}\geq g\}\mbox{~~~~~for all $x\in I$.}\]
\item
A point $x\in I$ is in the optimal stopping set $\{y:v(y)=g(y)\}$ if and only if there exists $c\in [a,b]$ such that
\[x\in\argmax \frac{g}{h_c}.\]
\end{enumerate}
\end{theorem}

\begin{theorem}\label{thm:worst-case_diff}
Let $x\in \R$ and let $c$ be a minimizer as in Theorem \ref{thm:main_diffusion} (i). Then $P^c$ is a worst-case measure for the process started in $x$.
\end{theorem}
\subsection{Example: An optimal decision problem for Brownian motions with drift}

The following example illustrates that our method also works in the case of a discontinuous reward function $g$, where differential equation techniques cannot be applied immediately. Furthermore, we see that our approach can be used for all parameters in the parameter space, although the structure of the solution changes.\\
Let $X=\sigma W_t+\mu t$ denote a Brownian motion with drift $\mu\in(-\infty,\infty)$ and volatility $\sigma$ under $P^0$, and let 
\[g(x)=\begin{cases}
1,&\;\;x\leq 0,\\
x,&\;\;x>0.
\end{cases}\]
The fundamental solutions are given by
\[\psi_+^\kappa(x)=e^{\alpha_1x},\;\psi_-^\kappa(x)=e^{\alpha_2x},\;\psi_+^{-\kappa}(x)=e^{\beta_1x},\;\psi_-^{-\kappa}(x)=e^{\beta_2x},\]
where $\alpha_1<0<\alpha_2$ and  $\beta_1<0<\beta_2$ are the roots of
\[1/2\sigma^2 z^2+(\mu-\kappa) z-r=0,\mbox{ resp. }1/2\sigma^2 z^2+(\mu+\kappa) z-r=0.\]
Using equation \eqref{eq:gamma} we obtain
\[h_c(x)=\begin{cases}
  \frac{\alpha_2}{\alpha_2-\alpha_1}e^{\alpha_1(x-c)}-\frac{\alpha_1}{\alpha_2-\alpha_1}e^{\alpha_2(x-c)},  & \text{if}~~x> c,\\
  \frac{\beta_2}{\beta_2-\beta_1}e^{\beta_1(x-c)}-\frac{\beta_1}{\beta_2-\beta_1}e^{\beta_2(x-c)},  & \text{if}~~x\leq c.
\end{cases}\]
We consider
\[l_*(c):=\sup_{y\leq 0}\frac{1}{h_c(y)}=\begin{cases}
\frac{1}{h_c(0)},&\;c\geq0,\\
1,&\;c\leq 0
\end{cases}\]
and $l^*(c):=\sup_{y\geq 0}\frac{y}{h_c(y)}=\frac{y_c}{h_c(y_c)}$, where $y_c$ denotes the unique maximum point of $y/h_c(y),y\geq 0$.\\
We first consider the case\ $1=l_*(0)\geq l^*(0)$. By Theorem \ref{thm:main_diffusion} (ii), we obtain that $x=0$ is in the optimal stopping set $S$ as a maximizer of $y\mapsto g(y)/h_0(y)$. Furthermore, by decreasing $c$ to $-\infty$, we see that $(-\infty,0]\subseteq S$. Since $l_*(c)\rightarrow 0$ and $l^*(c)\rightarrow \infty$ for $c\rightarrow\infty$, there exists a unique $c^*\geq 0$ such that $l_*(c^*)=l^*(c^*)$. Therefore, by Theorem \ref{thm:main_diffusion} (ii) again, $x^*:=y_{c^*}\in S$ and by increasing $c$ to $\infty$, we obtain that $S=(-\infty,0]\cup [x^*,\infty)$. Theorem \ref{thm:main_diffusion} (i) yields
\[v(x)=\begin{cases}
1,&\;x\leq 0\\
\frac{h_{c^*}(x)}{h_{c^*}(0)},&\;x\in[0,x^*],\\
x&\;x\geq x^*,
\end{cases}\]
see Figure \ref{amb1} below.
By Theorem \ref{thm:worst-case_diff} we furthermore obtain that $P^{c^*}$ is a worst-case measure for the process started in $x\in(0,x^*)$. That is, under the worst-case measure, the process has drift $\mu+\kappa$ on $[0,c^*)$ and drift $\mu-\kappa$ on $[c^*,x^*]$.\\
Now, we consider the case $1=l_*(0)< l^*(0)$. By a similar reasoning as in the first case, we see that there exists $c^*<0$ such that $l_*(c^*)=l^*(c^*)$. Write $x_*=c^*<0,x^*=y_{c^*}$. Then, $S=(-\infty,x_*]\cup[x^*,\infty)$ is the optimal stopping set and the value function is given by
\[v(x)=\begin{cases}
1,&\;x\leq x^*\\
{h_{c^*}(x)},&\;x\in[x_*,x^*],\\
x&\;x\geq x^*,
\end{cases}\]
see Figure \ref{amb2} below. The worst-case measure is given by $P^{c^*}$, which means that the process has drift $\mu-\kappa$ on $[x_*,x^*]$.

\begin{figure}[h]
\begin{minipage}{0.5\textwidth}
\begin{center}
\includegraphics[width=4.5cm]{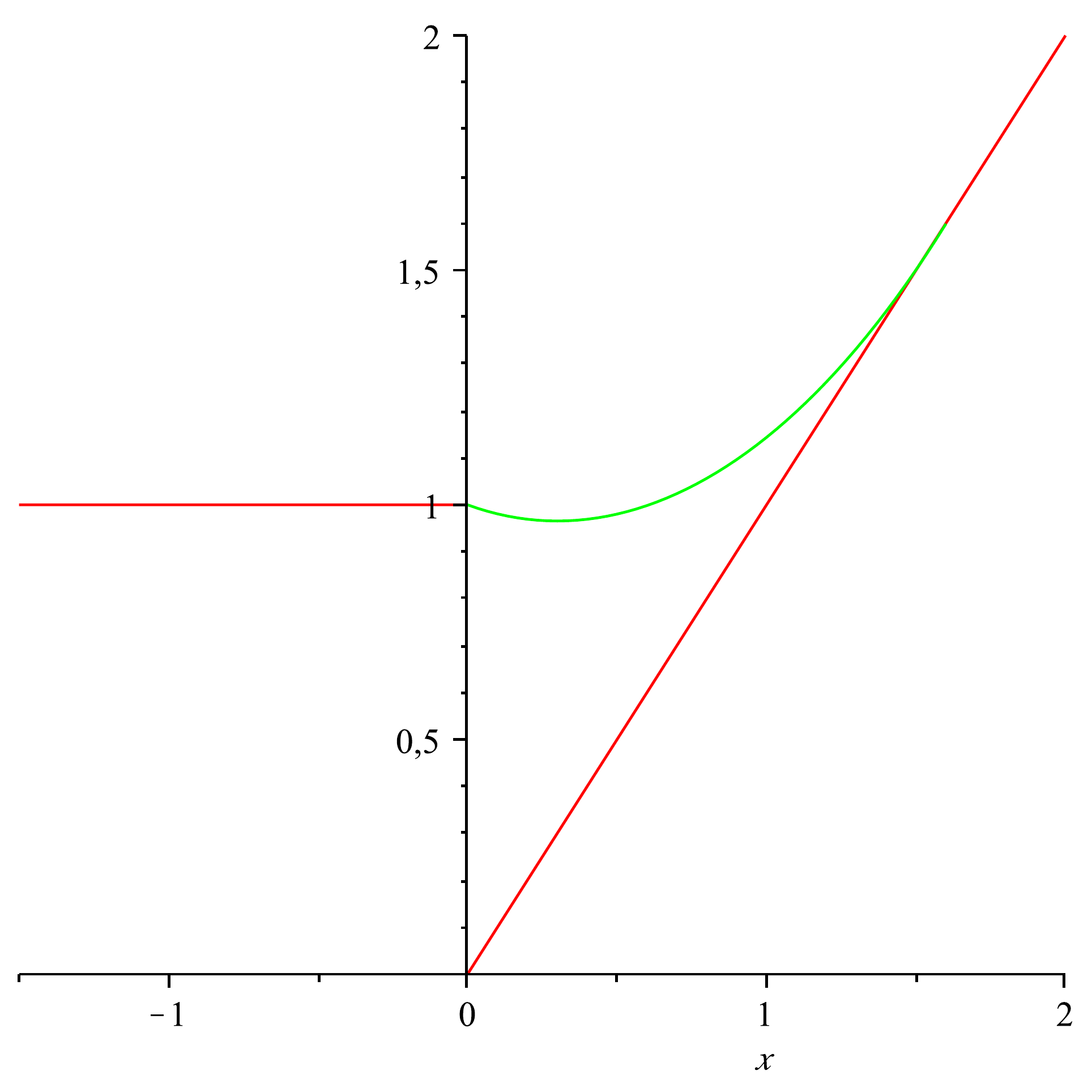}
\caption{Value function for $l_*(0)\geq l^*(0)$}\label{amb1}
\end{center}
\end{minipage}
\hfill
\begin{minipage}{0.5\textwidth}
\begin{center}
\includegraphics[width=4.5cm]{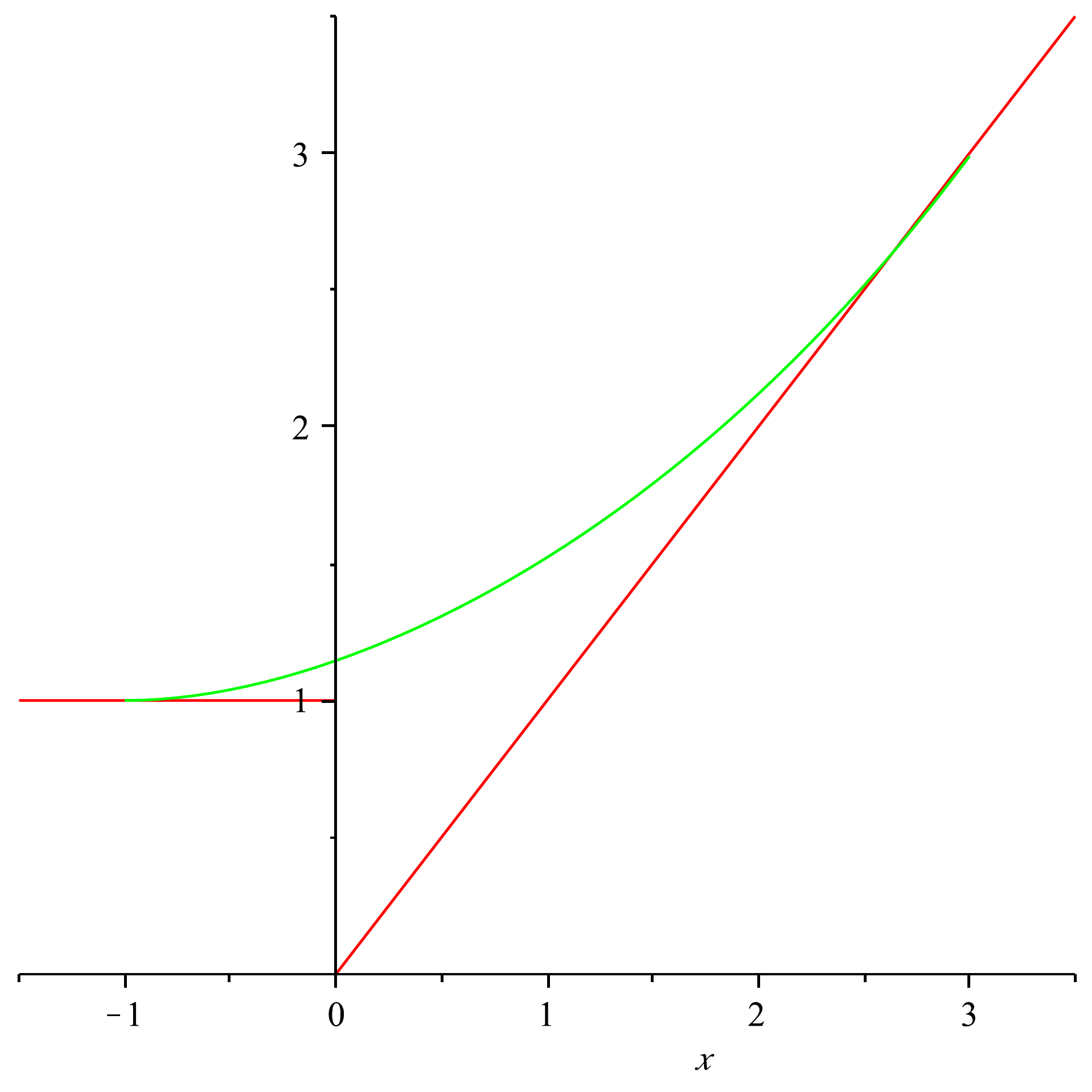}
\caption{Value function for $l_*(0)< l^*(0)$}\label{amb2}
\end{center}
\end{minipage}
\end{figure}

\section{Optimal decision for models with crashes} \label{sec:crashes}
Now denote by $Y$ a one-dimensional regular diffusion process on an interval $I$. Denote by $\mathcal{F}$ the natural filtration generated by $Y$. In this section we assume that all parameters of this process are known. This process represents the asset price process of the underlying asset if no crash occurs; therefore for economical plausibility it is reasonable to assume $I=(0,\infty)$.

Now we modify the process such that at a certain random time point $\sigma$ a crash of bounded height occurs.
To be more precise, let $c\in(0,1)$ be a given constant that describes an upper bound for the height of the crash.  For a given stopping time $\sigma$ and an $\mathcal{F}_\sigma$-measurable and $[c,1]$-valued random variable $\zeta$ we consider the modified process $X^{\sigma,\zeta}$ given by
\[X^{\sigma,\zeta}_t=\begin{cases}
 Y_t & t\leq\sigma \\
 \zeta Y_t &  t> \sigma.
\end{cases}\]
Now we consider the optimal stopping problem connected to the pricing of perpetual American options in this market, i.e., let $g:(0,\infty)\rightarrow[0,\infty)$ be a continuous reward function. We furthermore assume $g$ to be non-decreasing, so that a crash always leads to a lower payoff. We fix a constant discounting rate $r>0$ and furthermore assume that the holder of the option does know that the process will crash once in the future. We assume the crash to be observable for the decision maker, so she will specify her action by a pre-crash stopping time $\underline{\tau}$ and a post-crash stopping time $\overline{\tau}$, i.e. given $\sigma$ she takes the strategy
\begin{equation}\label{eq:strategy}
\tau=\tau_\sigma=\begin{cases}
 \underline{\tau}, &  \underline{\tau}\leq\sigma\\
\sigma+ \overline{\tau}\circ\theta_\sigma, & \mbox{ else},
\end{cases}
\end{equation}
where $\theta_\cdot$ denotes the time-shift operator.
As before we assume the holder of the option to be ambiguity averse in the sense that she maximizes her expected reward under the worst-case scenario, i.e. she tries to solve the problem
\begin{equation}\label{eq:value_crashes}
v(x)=\sup_{\underline{\tau},\overline{\tau}}\inf_{\sigma,\zeta}\E_x(e^{-r\tau}g(X^{\sigma,\zeta}_\tau)),
\end{equation}
where $\tau=\tau_\sigma$ is given as in \eqref{eq:strategy}.
\begin{remark}\label{rem:c}
Obviously by the monotonicity of the reward function we always have 
\[v(x)=\sup_{\underline{\tau},\overline{\tau}}\inf_{\sigma}\E_x(e^{-r\tau}g(X^{\sigma}_\tau))\]
where $X^{\sigma}:=X^{\sigma,c}$.
\end{remark}
We obtain the following reduction of the optimal stopping problem under ambiguity about the crashes: It shows that the problem can be reduced into one optimal stopping problem and one Dynkin game for the diffusion process $Y$ (without crashes).

\begin{theorem}\label{thm:reduction}
\begin{enumerate}[(i)]
\item Let $\hat{g}$ be the value function for the optimal stopping problem for $cY$ with reward $g$, i.e.
\begin{equation}\label{eq:stop}
\hat{g}(y)=\sup_{{\overline{\tau}}}\E_y(e^{-r{\overline{\tau}}}g(cY_{\overline{\tau}}))\mbox{~for all~}y\in(0,\infty)
\end{equation} 
and let $\hat{g}<\infty$. Then it holds that 
\begin{equation}\label{eq:game}
v(x)= \sup_{\underline{\tau}}\inf_\sigma\E_x(e^{-r\underline{\tau}}g(Y_{\underline{\tau}})\mathds{1}_{\{{\underline{\tau}}\leq \sigma\}}+e^{-r\sigma}\hat{g}(Y_\sigma)\mathds{1}_{\{{\underline{\tau}}>\sigma\}}).
\end{equation}
\item If $\overline{\tau}$ is optimal for \eqref{eq:stop} and $\underline{\tau},\sigma$ is a Nash-equilibrium for \eqref{eq:game}, then $(\un{\tau},\ov{\tau})$, $(\sigma,c)$ is a Nash-equilibrium for \eqref{eq:value_crashes}.
\end{enumerate}
\end{theorem}

\begin{proof}
\begin{enumerate}[(i)]
\item First fix $\un{\tau}$, $\ov{\tau}$. Then for all $\sigma$ by conditioning on $\mathcal{F}_\sigma$ we obtain
\begin{align*}
\E_x(e^{-r{\tau}}g(X^{\sigma}_\tau))=\E_x(&e^{-r\un{\tau}}g(Y_{\un{\tau}})\mathds{1}_{\{\un{\tau}\leq \sigma\}}+e^{-r(\sigma+\ov{\tau}\circ\theta_\sigma)}g(cY_{\sigma+\ov{\tau}\circ\theta_\sigma})\mathds{1}_{\{\un{\tau}>\sigma\}})\\
=\E_x(&e^{-r\un{\tau}}g(Y_{\un{\tau}})\mathds{1}_{\{\un{\tau}\leq\sigma\}}\\
&+e^{-r\sigma}\E_x(e^{-r(\ov{\tau}\circ\theta_\sigma)}g(cY_{\sigma+\ov{\tau}\circ\theta_\sigma})|\mathcal{F}_\sigma)\mathds{1}_{\{\un{\tau}>\sigma\}}).
\end{align*}
By the strong Markov property we furthermore obtain 
\[\E_x(e^{-r(\ov{\tau}\circ\theta_\sigma)}g(cY_{\sigma+\ov{\tau}\circ\theta_\sigma})|\mathcal{F}_\sigma)=\E_{Y_\sigma}(e^{-r{\overline{\tau}}}g(cY_{\overline{\tau}}))\leq \hat{g}(Y_\sigma).\]
Therefore,
\begin{align*}
\E_x(e^{-r\tau}g(X^{\sigma}_\tau))\leq \E_x(e^{-r\un{\tau}}g(Y_{\un{\tau}})\mathds{1}_{\{\un{\tau}\leq\sigma\}}+e^{-r\sigma}\hat{g}(Y_\sigma)\mathds{1}_{\{\un{\tau}>\sigma\}}),
\end{align*}
showing that
\begin{equation*}
v(x)\leq \sup_{\underline{\tau}}\inf_\sigma\E_x(e^{-r\underline{\tau}}g(Y_{\underline{\tau}})\mathds{1}_{\{{\underline{\tau}}\leq \sigma\}}+e^{-r\sigma}\hat{g}(Y_\sigma)\mathds{1}_{\{{\underline{\tau}}>\sigma\}}).
\end{equation*}
Now take a sequence of $1/n$-optimal stopping times $(\ov{\tau}_n)_{n\in\N}$ for the problem \eqref{eq:stop}, i.e. 
\[\hat{g}(y)\leq \E_y(e^{-r\ov{\tau}_n}g(cY_{\ov{\tau}_n}))+\frac{1}{n}\mbox{~~for all~$n\in\N$ and all $y$}.\]
Then
\[\E_{Y_\sigma}(e^{-r{\overline{\tau}_n}}g(cY_{\overline{\tau}_n}))\leq \hat{g}(Y_\sigma)+\frac{1}{n},\]
and hence considering the post-crash strategy $\ov{\tau}_n$ and arbitrary $\un{\tau},\sigma$ we see that 
\begin{equation*}
v(x)+\frac{1}{n}\geq \sup_{\underline{\tau}}\inf_\sigma\E_x(e^{-r\underline{\tau}}g(Y_{\underline{\tau}})\mathds{1}_{\{{\underline{\tau}}\leq \sigma\}}+e^{-r\sigma}\hat{g}(Y_\sigma)\mathds{1}_{\{{\underline{\tau}}>\sigma\}}),
\end{equation*}
proving equality.
\item is obvious by the proof of (i).
\end{enumerate}
\end{proof}

\begin{remark}
Note that the arguments used so far have nothing to do with diffusion processes, but can be applied in the same way for general \textit{nice} one-dimensional strong Markov processes, like one-dimensional Hunt processes. Nonetheless we decided to consider this more special setup because of its special importance and since the theory for explicitly solving optimal stopping problems and Dynkin games is well established. 
\end{remark}

The previous reduction theorem solves the optimal stopping problem \eqref{eq:value_crashes} since both problems \eqref{eq:game} and \eqref{eq:stop} are well-studied for diffusion processes, see e.g. the references given above for optimal stopping problems and \cite{EV}, \cite{A}, and \cite{Pe} for Dynkin games. It is interesting to see that the optimal stopping problem under crash-scenarios naturally leads to Dynkin games, which were studied extensively in the last years. The financial applications studied so far were based on Israeli options, which are (at least at first glance) of a different nature, see \cite{Ki}.

\subsection{Example: Call-like problem with crashes}
As an example we consider a geometric Brownian motion given by the dynamics
\[dX_t=X_t(\mu dt +\sigma dW_t),~~t\geq 0\]
and we take $g:(0,\infty)\rightarrow[0,\infty)$ given by $g(x)=(x-K)^+$, where $K>0$ is a constant. To exclude trivial cases we assume that $\mu<r$. Then a closed-form solution of the optimal stopping problem 
\begin{equation*}
\hat{g}(y)=\sup_{\tau}\E_y(e^{-r\tau}(cY_\tau-K)^+)=\sup_{\tau}\E_{cy}(e^{-r\tau}(Y_\tau-K)^+)
\end{equation*} 
is well known (see e.g. \cite[Chapter VII]{PS}) and is given by
\begin{equation*}
\hat{g}(y)=\begin{cases}
  cy-K,  & cy \geq x^*,\\
  d(cy)^\gamma, & cy<x^*,
\end{cases}
\end{equation*}
where $\gamma$ is the positive solution to 
\[\frac{\sigma^2}{2} z^2+\big(\mu-\frac{\sigma^2}{2}\big)z+r=0,\]
and $x^*$ and $d$ are appropriate constants. Moreover, the optimal stopping time is given by $\ov{\tau}:=\inf\{t\geq 0: X_t\geq x^*\}$. By Theorem \ref{thm:reduction} we are faced with the Dynkin game
\begin{equation}\label{eq:crashes_call} 
v(x)=\sup_{\un{\tau}}\inf_{\sigma}\E_x(e^{-r\un{\tau}}(Y_\un{\tau}-K)^+\mathds{1}_{\{\un{\tau}\leq \sigma\}}+e^{-r\sigma}\hat{g}(Y_\sigma)\mathds{1}_{\{\un{\tau}>\sigma\}}).
\end{equation}
To solve this problem first note that there exists $x'\in(K,x^*/c)$ such that $g(x)\leq \hat{g}(x)$ for $x\in(0,x']$ and $g(x)\geq \hat{g}(x)$ for $x\in[x',\infty)$; indeed, $x'$ is the unique positive solution to 
\[d(cx)^\gamma=x-K,\]
see Figure \ref{fig:call_crashes}.

\begin{figure}%
\begin{center}
\includegraphics[width=0.5\columnwidth]{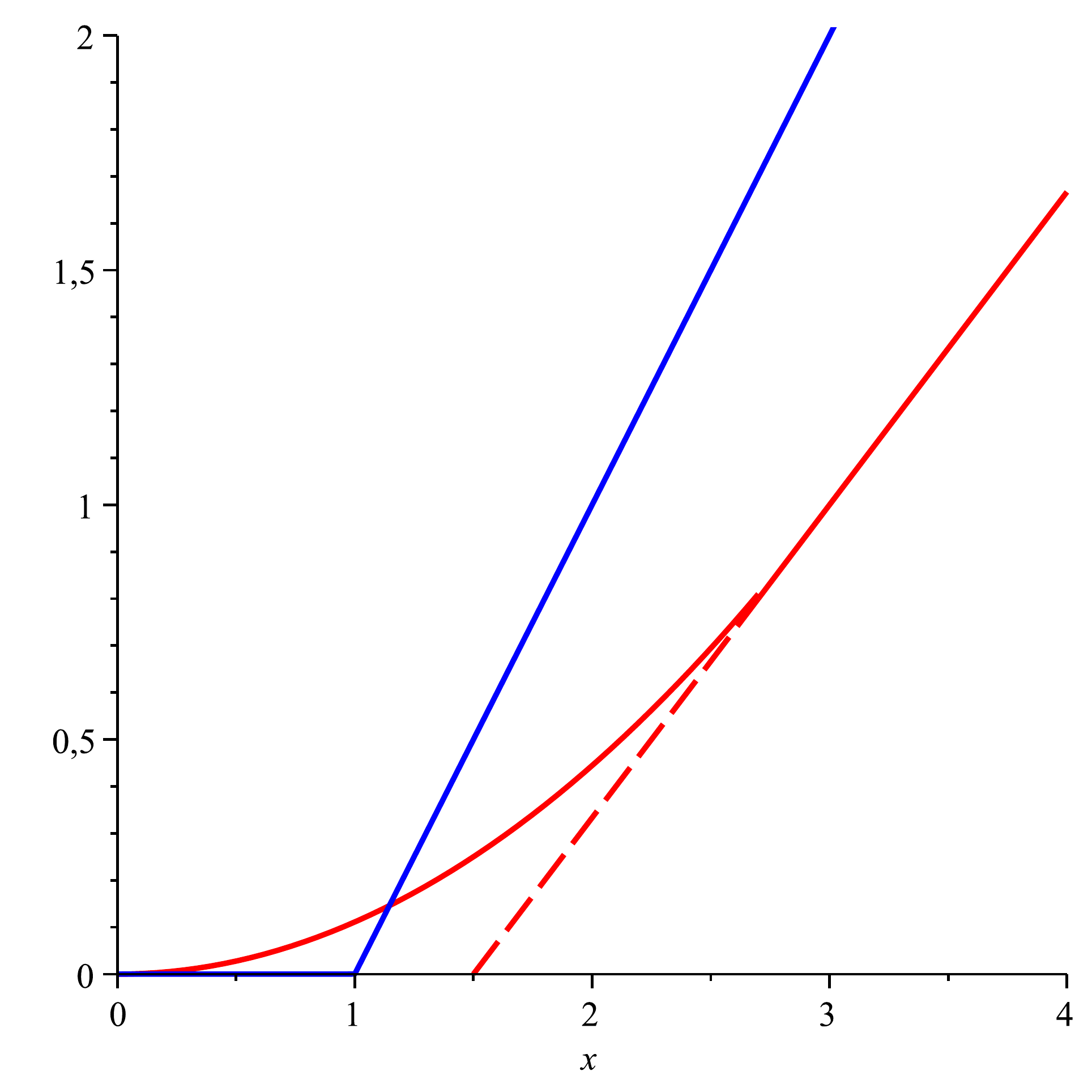}%
\caption{Graphs of $g$ (blue) and $\hat{g}$ (red).}\label{fig:call_crashes}
\end{center}
\end{figure}

We could use the general theory to solve the optimal stopping game \eqref{eq:crashes_call}, but we can also solve it elementary here:\\
First let $x>x'$. Then for all stopping times $\sigma^*$ with $\sigma^*=0$ under $P(\cdot | Y_0=x)$ and each stopping time $\un{\tau}$ we obtain 
\begin{align*}
&\E_x(e^{-r\un{\tau}}(Y_\un{\tau}-K)^+\mathds{1}_{\{\un{\tau}\leq \sigma^*\}}+e^{-r\sigma^*}\hat{g}(Y_{\sigma^*})\mathds{1}_{\{\un{\tau}>\sigma^*\}})\\
=&\E_x(g(x)\mathds{1}_{\{\un{\tau}=0\}}+\hat{g}(x)\mathds{1}_{\{\un{\tau}>0\}})\\
\leq& g(x)
\end{align*}
with equality if $\tau=0$ $P(\cdot | Y_0=x)$-a.s. On the other hand for $\tau^*=0$ the payoff is $g(x)$, independent of $\sigma$. \\

For $x\leq x'$ by by taking $\un{\tau}= \inf\{t\geq 0: X_t\geq x'\}$ we have for each stopping time $\sigma$ by definition of $x'$
\begin{align*}
&\E_x(e^{-r\un{\tau}}(Y_{\un{\tau}}-K)^+\mathds{1}_{\{\un{\tau}\leq \sigma\}}+e^{-r\sigma}\hat{g}(Y_\sigma)\mathds{1}_{\{\un{\tau}>\sigma\}})\\
=&\E_x(e^{-r\un{\tau}}(x'-K)\mathds{1}_{\{\un{\tau}\leq \sigma\}}+e^{-r\sigma}d(cY_\sigma)^\gamma\mathds{1}_{\{\un{\tau}>\sigma\}})\\
=&\E_x(e^{-r\un{\tau}}d(cx')^\gamma\mathds{1}_{\{\un{\tau}\leq \sigma\}}+e^{-r\sigma}d(cY_\sigma)^\gamma\mathds{1}_{\{\un{\tau}>\sigma\}})\\
=&\E_x(e^{-r(\sigma\wedge\un{\tau})}d(cY_{\sigma\wedge\un{\tau}})^\gamma)\\
=&dc^\gamma\E_x(e^{-r(\sigma\wedge\un{\tau})}(Y_{\sigma\wedge\un{\tau}})^\gamma)\\
=&dc^\gamma x^\gamma=\hat{g}(x),
\end{align*}
where the last equality holds by the fundamental properties of the minimal $r$-harmonic functions, see e.g. \cite[II.9]{BS}. By taking any stopping time $\sigma^*\leq \inf\{t\geq0:y_t\geq x'\}$ and any stopping time $\tau$ the same calculation holds.\\
Putting pieces together we obtain that $\un{\tau},\sigma^*$ is a Nash-equilibrium of the Dynkin game \eqref{eq:crashes_call} for any stopping time $\sigma^*\leq \un{\tau}$.\\
By applying Theorem \ref{thm:reduction} we get 
\begin{proposition}
The value function $v$ is given by
\begin{equation*}
v(x)=\begin{cases}
  x-K,  & x\geq x',\\
  d(cx)^\gamma, & x<x',
\end{cases}
\end{equation*}
and for 
\[\un{\tau}=\inf\{t\geq 0: X_t\geq x'\}\hspace{2cm} \mbox{'pre-crash strategy'}\]
and 
\[\ov{\tau}=\inf\{t\geq 0: X_t\geq x^*\}\hspace{2cm} \mbox{'post-crash strategy'}\]
and any stopping time $\sigma^*\leq \tau$ it holds that $(\un{\tau},\ov{\tau})$, $(\sigma^*,c)$ is a Nash-equilibrium of the problem. 
\end{proposition}
The solution to this example is very natural: If the investor expects a crash in the market, then she exercises the option as soon as the asset price reaches the level $x'$ (pre-crash strategy). After the crash, i.e. if the investor does not expect to have more crashes, then she takes the ordinary stopping time, i.e. she stops if the process reaches level $x^*>x'$ (post-crash strategy).

\section*{Acknowledgements}
This paper was partly written during a stay at \AA bo Akademi in the project \textit{Applied Markov processes -- fluid queues, optimal stopping and population dynamics, Project number: 127719  (Academy of Finland)}.  I would like to express my gratitude for the hospitality and support.

\bibliographystyle{alpha}
\bibliography{ambiguity}

\end{document}